\documentclass[conference]{IEEEtran}

\usepackage{graphicx}
\DeclareGraphicsExtensions{.pdf,.jpeg,.png}
\usepackage{epstopdf}
\graphicspath{ {./images/} }

\usepackage[cmex10]{amsmath}

\usepackage{algorithmic}
\usepackage{algorithm}

\usepackage{array}
\usepackage{mdwmath}
\usepackage{mdwtab}
\usepackage{eqparbox}

\ifCLASSOPTIONcompsoc
\usepackage[caption=false,font=normalsize,labelfont=sf,textfont=sf]{subfig}
\else
\usepackage[caption=false,font=footnotesize]{subfig}
\fi

\usepackage{url}

\usepackage{multirow}
\usepackage{color}
\usepackage{breqn}
\usepackage{dsfont}
\usepackage{mathrsfs}
\usepackage[T1]{fontenc}
\usepackage{paralist}

\newtheorem{theorem}{Theorem}[section]

\newtheorem{proposition}[theorem]{Proposition}

\newenvironment{proof}[1][Proof]{\begin{trivlist}
\item[\hskip \labelsep {\bfseries #1}]}{\end{trivlist}}
\newenvironment{definition}[1][Definition]{\begin{trivlist}
\item[\hskip \labelsep {\bfseries #1}]}{\end{trivlist}}

\newenvironment{remark}[1][Remark]{\begin{trivlist}
\item[\hskip \labelsep {\bfseries #1}]}{\end{trivlist}}

\usepackage{enumerate}
\usepackage{amssymb}
\usepackage{setspace}

\makeatletter
\def\ps@headings{%
\def\@oddhead{\mbox{}\scriptsize\rightmark \hfil \thepage}%
\def\@evenhead{\scriptsize\thepage \hfil \leftmark\mbox{}}%
\def\@oddfoot{}%
\def\@evenfoot{}}
\makeatother
\renewcommand{\baselinestretch}{0.98}

\begin{document}
\title{Distributed Association Control and Relaying in\\ MillimeterWave Wireless Access Networks\vspace{-0.5ex}}
\author{\IEEEauthorblockN{Yuzhe Xu, George Athanasiou, Carlo Fischione}
\IEEEauthorblockA{
KTH Royal Institute of Technology, 
Stockholm, Sweden\\
Email: \{yuzhe, georgioa, carlofi\}@kth.se\vspace{-4ex}}
\and
\IEEEauthorblockN{Leandros Tassiulas}
\IEEEauthorblockA{
University of Thessaly, Volos, Greece\\
E-mail: leandros@uth.gr\vspace{-4ex}}}
\maketitle
\begin{abstract} 
In millimeterWave wireless networks the rapidly varying wireless channels demand fast and dynamic resource allocation mechanisms. This challenge is hereby addressed by a distributed approach that optimally solves the fundamental resource allocation problem of joint client association and relaying. The problem is posed as a multi-assignment optimization, for which a novel solution method is established by a series of transformations that lead to a tractable minimum cost flow problem. The method allows to design distributed auction solution algorithms where the clients and relays act asynchronously. The computational complexity of the new algorithms is much better than centralized general-purpose solvers. It is shown that the algorithms always converge to a solution that maximizes the total network throughput within a desired bound. Both theoretical
and numerical results evince numerous useful properties
in comparison to standard approaches and
the potential applications to forthcoming millimeterWave wireless
access networks.
\end{abstract}

\section{Introduction}\vspace{-0.05in}
MillimeterWave~(mmW) communications in the 60 GHz frequency band have recently attracted the interest of academia, industry, and standardization bodies due to the low-cost mmW radio frequency integrated circuit design~\cite{Giannetti1999, Smulders2007, Doan2004}. Unlicensed continuous spectrum of 7~GHz is available in many countries worldwide~\cite{Rappaport2011}. Therefore, mmW systems became attractive for gigabit wireless applications such as large file transfer, wireless docking stations, wireless gigabit ethernet, uncompressed high definition video transmission. MmW radio access technology can enable dense and small-cells~\cite{Hoydis11} and mobile data offloading~\cite{Lee12} scenarios, which are nowadays strongly motivated by the increased end-user connectivity requirements and mobile traffic.

MillimeterWave communications can provide data rates in the range of several Gbps~\cite{Daniels10}. In contrast to the existing radio frequency technologies, such as the 2.4~GHz or 5~GHz band (WiFi), the interference levels for mmW networks are low~\cite{Rappaport2011}. The compact design of mmW radio allows the use of multiple-antenna solutions at the user device. However, several challenges arise when mmW technology is applied to support the user demands. First, the severe channel attenuation and the high path loss require high transmission power levels. The oxygen absorption loss is in the range of 5-30~dB/km and the penetration loss can be too high through typical building's materials~\cite{Daniels10}. As a consequence, line-of-sight (LOS) is required in most of the cases, while signals along the second order (reflections, etc) are highly attenuated and too weak. Therefore, service disruptions are possible due to non-line-of-sight (NLOS) or blockage effects~\cite{Singh11}. The exploitation of these unique characteristics is essential for efficient resource allocation algorithms.   

%

In this paper we consider a typical wireless access problem, where each client has to be associated to one wireless access points (APs) and where relay clients cooperatively assist other clients in case of communication blockage, NLOS or very poor link performance. 
Such a problem is more challenging in mmW band than in traditional wireless networks since the wireless channel is unstable and several events hinder the efficient operation of the network, such as moving obstacles that block the communication~\cite{Singh07}. This demands the use of relay clients that may serve other clients in case that the direct link to AP is not efficient. Specifically, we consider a challenging joint client association and relaying problem, where the objective is to \emph{maximize the total throughput} that the clients achieve. 

The joint wireless access association and relaying problem is posed as a multi-assignment optimization whose challenge is the unavailability of efficient distributed solution methods. The use of existing centralized solvers would not be possible due to the fast variations of the wireless channels in mmW that cause dramatic block of the communication at shorter time scales than the time needed for the central solver to work. Therefore, we propose a novel solution method that converts the initial problem to a minimum cost flow problem, which allows us to design two efficient solvers by a combination of \emph{distributed auction algorithms}, one for static networks and one for dynamic networks. Our algorithms exploit the specific network optimization structure of the problem, and thus they are much more powerful than computationally intensive general-purpose solvers~\cite{Bertsekas1998} for multi-assignment problems. New fundamental theoretical results are established for the optimality, convergence, and the complexity  of the distributed algorithms. In addition, numerical results evince numerous properties in comparison to existing approaches.

The rest of paper is organized as follows. In Section~\ref{sec.related-work} we give a literature overview. In Section~\ref{sec.system-model}, we describe the system model and problem formulation. The solution method is presented in Section~\ref{sec.solution}. The distributed auction algorithms is introduced in Section~\ref{distr_algo}. Then, the numerical results are presented in Section~\ref{sec.numerical-examples}. Lastly, Section~\ref{sec.conclusion} concludes the paper.
\section{Related Work}\vspace{-0.05in}
\label{sec.related-work}

Resource allocation for wireless local area networks has been the focus of intense research. The performance analysis of the basic client association policy that IEEE 802.11 standard defines, based on the received signal strength indicator (RSSI), was performed in several studies. It was shown that this basic association policy can lead to inefficient use of the network resources~\cite{Bejerano2}, since RSSI is not an efficient decision metric for user association. High RSSI values cannot univocally indicate high throughput. 
Therefore, better client association policies have been proposed\cite{Athanasiou07, Athanasiou09, Shakkottai}. 

The current 60 GHz mmW standardization bodies, such as IEEE 802.11ad~\cite{802_11ad} and IEEE 802.15.3c~\cite{802_15_3c}, support the RSSI-based mechanism as the basic association/handoff functionality. The accuracy of the RSSI-based technique is significantly affected by the high path loss, dispersion, and directionality of the mmW wireless channel. New metrics that better reflect these characteristics are required. The existing approaches are excellent for 2.4, 5 GHz communication, but are unfortunately hard to apply in mmW networks due to the above mentioned characteristics. The mmW channel has significant differences compared to the rest of wireless access technologies~\cite{Qiao11, Singh11, Genc12, Lin12}, namely, severe channel attenuations, high path loss, directionality, and blockage. It follows that mmW networks demand novel mechanisms for resource allocation. Our previous approach \cite{Athanasiou-etal-2013} was the first to study the client association in 60 GHz mmW wireless access networks. However, the focus was on the network performance (achieving load balancing) and not on optimizing the total throughput of the clients. Moreover, we did not consider relaying clients, which substantially increases the difficulty of the problem. 

Client relaying allows to overcome performance anomalies in IEEE 802.11 wireless local area networks~\cite{Heusse05}. A typical relaying scenario considers intermediate clients to relay the traffic for low data rate clients \cite{Bahl09}. The mechanism that decides when a client's traffic should be relayed, and to whom, plays an essential role. Several approaches analyzed the relaying performance and proposed mechanisms that consider the estimated transmission rates, throughput, and delays in the communication with relay clients \cite{Bahl09, Narayanan07, Zhang10}. In mmW wireless access networks relaying is of greater usage than in traditional wireless access networks (such as  WiFi), since the wireless channel is unstable and several events can violate the efficient operation of the network, such as moving obstacles that can block the communication. Some recent approaches appeared in literature in mmW networks, as we survey below. 

In~\cite{Singh10}, a MAC protocol that employs memory and learning to address deafness, while exploiting the reduction of interference between simultaneous transmissions, is proposed. The directionality and blockage problems of mmW networks are studied in \cite{Singh07}. A cross-layer approach is presented, where a single hop transmission is preferred when line of sight (LOS) is available and a relay node is randomly selected as an alternative. In \cite{Cai2010-2}, a resource management mechanism is proposed based on the exclusive region (ER) to exploit the spatial reuse of mmW networks. In~\cite{Singh11}, the author propose an interference analysis framework that enables a quantitative evaluation of collision loss probability for a mmW mesh network with uncoordinated transmissions, as a function of the antenna patterns and spatial density of simultaneously transmitting nodes. Concurrent transmissions in 60 GHz wireless networks are studied in~\cite{Qiao11} by exploiting the spatial reuse and time division multiplexing gain. It is shown that the network throughput is improved compared to single hop transmission schemes. However, none of the previous approaches studied jointly the client association and relaying. We believe that these two problems are strongly interconnected and a joint solution is of great importance in mmW wireless networks. Moreover, there are no distributed and lightweight algorithms in literature that are able to efficiently solve these problems. 

This paper considers the special characteristics of the mmW channel in an optimization problem where the objective is to \emph{maximize the total throughput} that the clients achieve. To that goal, we design a combination of \emph{distributed auction algorithms} to jointly optimize the client association and relay selection processes. We believe that this paper is the first to study jointly the AP and relaying assignment, which are the two interconnected essential resource allocation problems in mmW wireless access networks. Our work is presented and evaluated in the forthcoming sections. 

\begin{figure}[t]\vspace{-0.3in}
  \centering
  \includegraphics[width=0.3\textwidth]{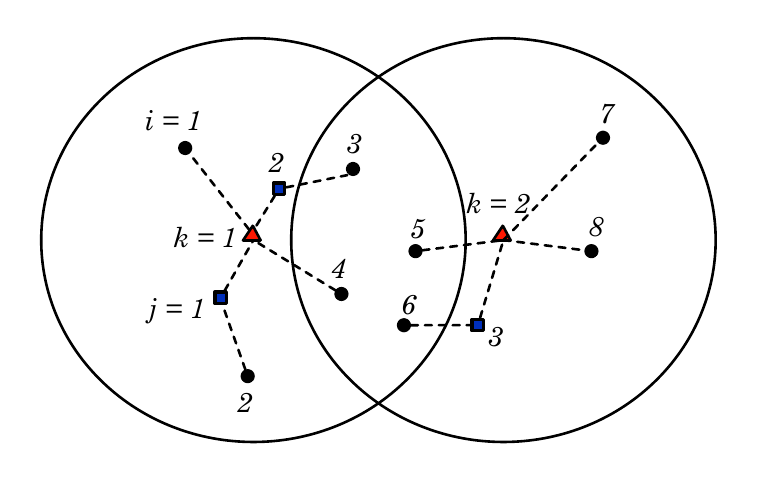}\vspace{-0.2in}
  \caption{Example of wireless access network with 8 clients, 3 relays, and 2 APs: $\mathcal M = \{1,\dots,8 \}$, $\mathcal N = \{1,2,3 \}$, and $\mathcal K = \{1 ,2\}$ are marked by circles, squares and triangles respectively. Consider the client 3,  we have $\mathcal N(i=3)=\{1,2,3\}$, $\mathcal K(i=3) = \{1,2\}$.  The dash lines represent the association. The pair client-relay-AP $(2,1,1)$ and client-AP $(1,1)$ are in assignment $\mathcal S$. }\vspace{-0.2in}
\label{fig network}
\end{figure}

\section{System Model And Problem Formulation}\vspace{-0.05in}
\label{sec.system-model}
Consider a mmW wireless network with $M$ clients and $K$ APs, where each client can be associated to one of the available APs. Each client may skip association with APs and establish a connection via one of the available $N\leq M$ relays, where a relay is a client that in addition to its transmissions help another client with an AP. We assume that every relay can serve only one client per time. We denote by $\mathcal M$, $\mathcal N$, and $\mathcal K$ the sets of clients, relays and APs respectively. The nonempty sets of the clients to which AP $k$, or the relay node $j$, can be connected are denoted by $\mathcal{M} (k)$ and $\mathcal{M} (j)$, respectively. Similarly, the nonempty sets of the APs by which client $i$ or relay node $j$ can be served are denoted by $\mathcal{K}(i)$ and $\mathcal{K} (j)$, while such sets of the relays are denoted by $\mathcal{N} (i)$ and $\mathcal{N} (k)$, respectively. $\mathcal A$ is the set of all possible client-AP pairs $(i,j)$ with $i\in \mathcal M(k)$ and client-relay-AP pairs $(i,j,k)$ with $i\in\mathcal M (j)$, $j\in \mathcal N k$. A feasible assignment $\mathcal{S}$ is defined as a subset of $\mathcal A$, where a relay can belong to one client-relay-AP pair, and a client must belong to either a client-relay-AP or a client-AP pair. Lastly,  assume the APs are placed uniformly, and the clients and relays are distributed uniformly at random inside the range of the APs. An example of this network is shown in Fig.~\ref{fig network}.

Each client, relay, and AP is equipped with steerable directional antennas and it can direct its beams to transmit or to receive~\cite{Cai2010-2}.
Therefore, every AP $k$ can support all the clients in $\mathcal{K} (i)$ and the relays in $\mathcal{K} (j)$ with a separate transmit beam. This basically gives a broadband interference and an additive white Gaussian noise channel. By the Friis transmission equation the received signal power can be calculated at distance $d$ between the transmitter and the receiver. It follows that the maximum achievable data rate at distance $d$ is
\begin{equation}\vspace{-0.07in}
\label{ch-capacity-in-d}
    R(d) = W \log_2 \left( 1 + \frac{P_{\rm T} G_{\rm R} G_{\rm T} \lambda^2 d_0^\eta}{16\pi^2(N_0 + I)W d^\eta}  \right)\,,
\end{equation}
where $P_{\rm T}$ is the transmitted power, and $G_{\rm T}$ and $G_{\rm R}$ are respectively the antenna gains of the transmitter and receiver. $\lambda$ is the wavelength, $d_0$ is the far field reference distance, and $\eta$ is the path loss exponent (typically $\eta \in [2,6]$ in IEEE 802.11ad networks~\cite{Suiyan2009, 802_11ad}). $W$ is the system bandwidth, $N_0$ is the power spectral of the noise, and $I$  is the broadband interference. The well studied 60~GHz propagation characteristics~\cite{Mudumbai2009}, such as highly directional transmissions with very narrow beamwidths and increased path losses due to the oxygen absorption, allow us to assume that the communication interference $I$ is very small and does not affect significantly the achievable communication rates.\footnote{In~\cite{Mudumbai2009}, a probabilistic analysis of the interference incurred in mmW networks has proven that even uncoordinated transmission for different transmit-receive pairs leads to small collision probabilities. Therefore, the links in the network can be considered as \emph{pseudo-wired}, and interference can essentially be ignored in MAC or higher layers design. Similar assumptions for mmW communications are also supported by using efficient scheduling algorithms for concurrent transmissions~\cite{Qiao11}.} 
Note that in mmW networks, the flow throughput decreases significantly over distance as modeled by Eq.~\eqref{ch-capacity-in-d}, indicating that short links, by the use of relays, may be preferred to achieve high transmission data rates.

When client $i$ is associated to AP $k$, we denote client $i$'s throughput benefit as $a_{(i,k)}$, whereas when client $i$ is associated to AP $k$ by an intermediate relay client $j$, we denote the client $i$'s throughput benefit as $a_{(i,j,k)}$. We let the value of the throughput benefit be $a_{(i,k)} = R(d_{ik})\,$, and $a_{(i,j,k)} = \min \left\{ R(d_{ij}),R(d_{jk})\right\}\,$. 
It follows that the link rate is bounded by the lowest rate when using a relay. 
To describe the client-AP or client-relay-AP association, we introduce the binary decision variables $x_{(i,k)}=1$ if $(i,k)\in \mathcal S$ and $x_{(i,k)}=0$ otherwise, for all $(i,k)\in \mathcal A$. Moreover, $x_{(i,j,k)}=1$ if $(i,j,k)\in\mathcal S$ and $x_{(i,j,k)}=0$ otherwise, for all $(i,j,k)\in \mathcal A$. Naturally, the total throughput benefit of the clients in the network is given by\vspace{-0.05in}
\begin{equation}\vspace{-0.05in}
\label{throughput}
    u =\sum_{(i,k)\in \mathcal A} a_{(i,k)} x_{(i,k)} + \sum_{(i,j,k)\in \mathcal A } a_{(i,j,k)} x_{(i,j,k)\,,} 
\end{equation}
Our goal is to find the optimal assignment $\mathcal S^*$ that maximizes $u$. This resource allocation problem can be formulated into a multi-dimensional assignment problem as
\begingroup
\thickmuskip=0.2\thickmuskip
\begin{subequations}\vspace{-0.06in}
    \label{multidim-assign-problem}
    \begin{align}
        \underset{x_{(i,k)}, \:\: x_{(i,j,k)}}{\textrm{maximize}}  &\ u \\ \label{client_constraint}
 {\rm s.t.}          &\sum_{(i,k)\in \mathcal A} x_{(i,k)} + \sum_{(i,j,k)\in\mathcal A} x_{(i,j,k)} =1, \  \forall i \in \mathcal M\,,\\
                        \label{relay_constraint}
                        & \sum_{(i,j,k) \in \mathcal A} x_{(i,j,k)} \leq 1, \ \forall j\in \mathcal N \,,\\
			\label{variable_constraint}
                        & x_{(i,j,k)}, x_{(i,k)} = \{ 0,1\}, \forall (i,j), (i,j,k) \in \mathcal A \,,
    \end{align}
\end{subequations}
\endgroup
where the known parameters are $a_{(i,k)}$ and $a_{(i,j,k)}$. Constraint~\eqref{client_constraint} assures that client $i$ is associated to one AP or connected to one relay. Constraint~\eqref{relay_constraint} ensures that relay $j$ can assist one client at most. Constraint~\eqref{variable_constraint} indicates that the decision variables are binary. Due to the limited available connections of the APs we must keep the link balance for APs (in terms of the number of connections) over the time. This introduces a time dimension and the consecutive solutions of problem~\eqref{multidim-assign-problem} must also satisfy the following constraint: 
\begin{equation}
\label{load_balance_contraint}
    \mathbf{E} \left[ \sum_{(i,k)\in \mathcal A} x_{(i,k)} + \sum_{(i,j,k)\in \mathcal A} x_{(i,j,k)}  \right] = \frac{M}{K} \,,\ \forall k\in \mathcal K\,,
\end{equation}
where the expectation is taken with respect to the random distribution of clients and relays over the time. The randomness is hidden in set $\mathcal A$. 

Problem~\eqref{multidim-assign-problem} is a special case of multi-assignment problems that in general have no closed form solutions and are NP-complete~\cite{Horst-Pardalos-Toai-00}. Moreover, the problem is combinatorial, and may have multiple optimal solutions. We have to rely on exponentially complex global methods~\cite{Horst-Pardalos-Toai-00} to solve it, unless new methods are developed. The computational cost for directly solving problem~\eqref{multidim-assign-problem} by searching is very high, since the number of possible combinations is $\mathcal{O}(M N! K^{M})$ when $M > N$. In what follows we propose a novel distributed solution approach based on auction algorithms.

\section{Centralized Solution Method}\vspace{-0.05in}
\label{sec.solution}

In this section, we present a novel solution method for problem~\eqref{multidim-assign-problem} in three steps: (a) We propose a set of intermediate transformations with some virtual entities; (b) We transfer problem~\eqref{multidim-assign-problem} into an equivalent asymmetric assignment problem; (c) Finally, we covert that asymmetric assignment problem into a minimum flow problem that is solved by a centralized auction-based approach. We present the details in the sequel. 
\subsection{Transformation to clients-objects domain}\vspace{-0.05in}
\label{sec.algorithms}
For every client $i$ and relay $j$, we construct new nonempty sets of APs $\mathcal K^*(i)$ and $\mathcal K^* (j)$, such that
\begin{align*}
    \mathcal K^*(i) &= \left\{ k_i^* \bigg| k_i^* = \underset{k\in \mathcal M(k)}{\textrm{argmax}}\, R(d_{ik})\right\} \,,\\
    \mathcal K^*(j) &= \left\{ k_j^* \bigg| k_j^* = \underset{k\in \mathcal N(k)}{\textrm{argmax}}\, R(d_{jk})\right\} \,.
\end{align*}
Then we have the following important result: 
\begin{proposition}
\label{prop.uniform_connection}
Consider the clients, relays and APs are distributed uniformly at random. Any feasible assignments for problem~\eqref{multidim-assign-problem}, using sets $\mathcal K^*(i)$ and $\mathcal K^*(j)$ instead of $\mathcal K(i)$ and  $\mathcal K(j)$, ensure constraint~\eqref{load_balance_contraint}.
\end{proposition}
\begin{proof}
At time $t$, let $\Pr (i, j)$ denote the probability of client $i$ being connected to relay $j$, which depends on both the positions of all the clients and relays, and the assignment strategy. Let $\Pr (j, k)$ denote the probability of relay $j$ being associated to AP $k$. Since the clients, relays, and APs are placed uniformly at random in this network, we have $\Pr (j, k)= 1/M$, and the conditional probability of client $i$ associated to AP $k$, $\Pr(i,k | i,{\rm AP}) = 1/K$, given there is no relay involving. Moreover, since all the clients are associated to either relays or APs under a feasible assignment of problem~\eqref{multidim-assign-problem}, $\sum_{j \in \mathcal  N} \Pr (i, j) + \Pr (i, \rm AP)= 1$. 
Thus the probability of client $i$ being associated to AP $k$ is
\begin{align*}
    \Pr (i, k)  =& \sum_{j\in\mathcal N} \Pr (i,j) \Pr (j, k) \\
    &+ \left(1 - \sum_{j\in\mathcal N} \Pr (i, j)\right) \Pr (i, k| i , {\rm AP} ) = \frac{1}{K}\,,
\end{align*}
which indicates that client $i$ and relay $j$ have the same probability $1/K$ to be connected to AP $k$. Thus the expectation number of clients and relays that connect to AP $k$ can be obtained as $M/K$, which completes the proof.
\end{proof}

Proposition~\ref{prop.uniform_connection} implies the link balancing for APs in the number of connections by using the reconstructed sets $\mathcal K^*(i)$ and $\mathcal K^* (j)$. In addition, for each client $i$ we define a new set $\mathcal Q(i)$ as
\begin{equation*}
    \mathcal Q(i) = \mathcal N^*(i)\bigcup \{N+i\}\,,
\end{equation*}
where
\begin{equation*}
	 \mathcal N^*(i) = \left\{ j \in \mathcal N(i):\min\{R (d_{ij}), R (d_{jk_j^*})\} > R(d_{ik_i^*}) \right\}\,,
\end{equation*}
and where $N+i$ is a virtual relay that is associated to AP $k_i^*$ as shown in Fig.~\ref{fig.example_transfer}. Moreover, we place virtual relay $N+i$ infinity close to client $i$.  Thus via relay $N+i$, the transmission rates from client $i$ to any APs are the same as those from client $i$ to APs directly. 

The elements in $\mathcal Q$ are named as ``objects''. To be consistent with auction theory terminology, in what follows the relays, APs, and objects are considered as same entities, and sometimes the term object is used for all of them. The set of clients that could be associated to object $q \in \mathcal Q$ is denoted by $\mathcal M(q)$. Note that relay $N+i$ is only accessible by client $i$, such that $\mathcal M (q=N+i) = \{i\}$.
\subsection{Asymmetric Assignment Problem}\vspace{-0.05in}
\label{prop.ass_assignment}
Using the previous derivations, problem~\eqref{multidim-assign-problem} can be transferred to a standard asymmetric assignment problem
\begin{subequations}
\label{asymmetric-assign-problem}
    \begin{align}
        \underset{y_{(i,q)}}{\textrm{maximize}}  & \quad \sum_{(i,q)\in\mathcal A^*} \beta_{(i,q)} y_{(i,q)} \\
        \label{asym_relay_constraint}
        {\rm s.t.}          & \quad \sum_{i\in \mathcal M (q)} y_{(i,q)} \leq 1 ,\ \forall q \in \mathcal Q(i) \,,\\
        \label{asym_client_constraint}
                            & \quad \sum_{q\in \mathcal Q(i)} y_{(i,q)} = 1,\ \forall i\in \mathcal M,\\
        \label{asym_variable_constraint}
                            & \quad y_{(i,q)} \geq 0 \,,\ \forall (i,q)\in\mathcal A^* \,,
    \end{align}
\end{subequations}
where set $\mathcal A^*$ contains all possible pairs, such that $i\in\mathcal M$ and $q\in \mathcal Q (i)$. The known parameter is $\beta_{(i,q)}$. It equals to $a_{(i,j,k_j^*)}$, if $q=1,\dots,N$, and $a_{(i,k_i^*)}$, otherwise. Constraint~\eqref{asym_relay_constraint} ensures that object $q$ can be assigned to one client at most. Constraint~\eqref{asym_client_constraint} ensures that client $i$ can be assigned to one and only one object. Object $q \in \{ N+1,\dots,N+M\}$ represents an AP, which are now represented as virtual objects and forced to be associated at most once. Constraint~\eqref{asym_variable_constraint} resulted from the relaxation of ~\eqref{variable_constraint}. This relaxation does not give a suboptimal solution because, based on the theory for assignment-like problems~\cite{Bertsekas1998} if problem~\eqref{asymmetric-assign-problem} has an optimal solution $y^*_{(i,q)}$, then $y^*_{(i,q)}=0$ or $1$, $\forall (i,q)\in\mathcal S^*$. Lastly, note that since for every client we can find at least a virtual relay, the number of objects is larger than the number of clients, from which the asymmetric assignment problem is obtained~\cite{Bertsekas1998}. The variables of problem~\eqref{multidim-assign-problem}  are recovered from those of~\eqref{asymmetric-assign-problem} by the following proposition. 

\begin{figure}\vspace{-0.3in}
\centering
\includegraphics[width=0.3\textwidth]{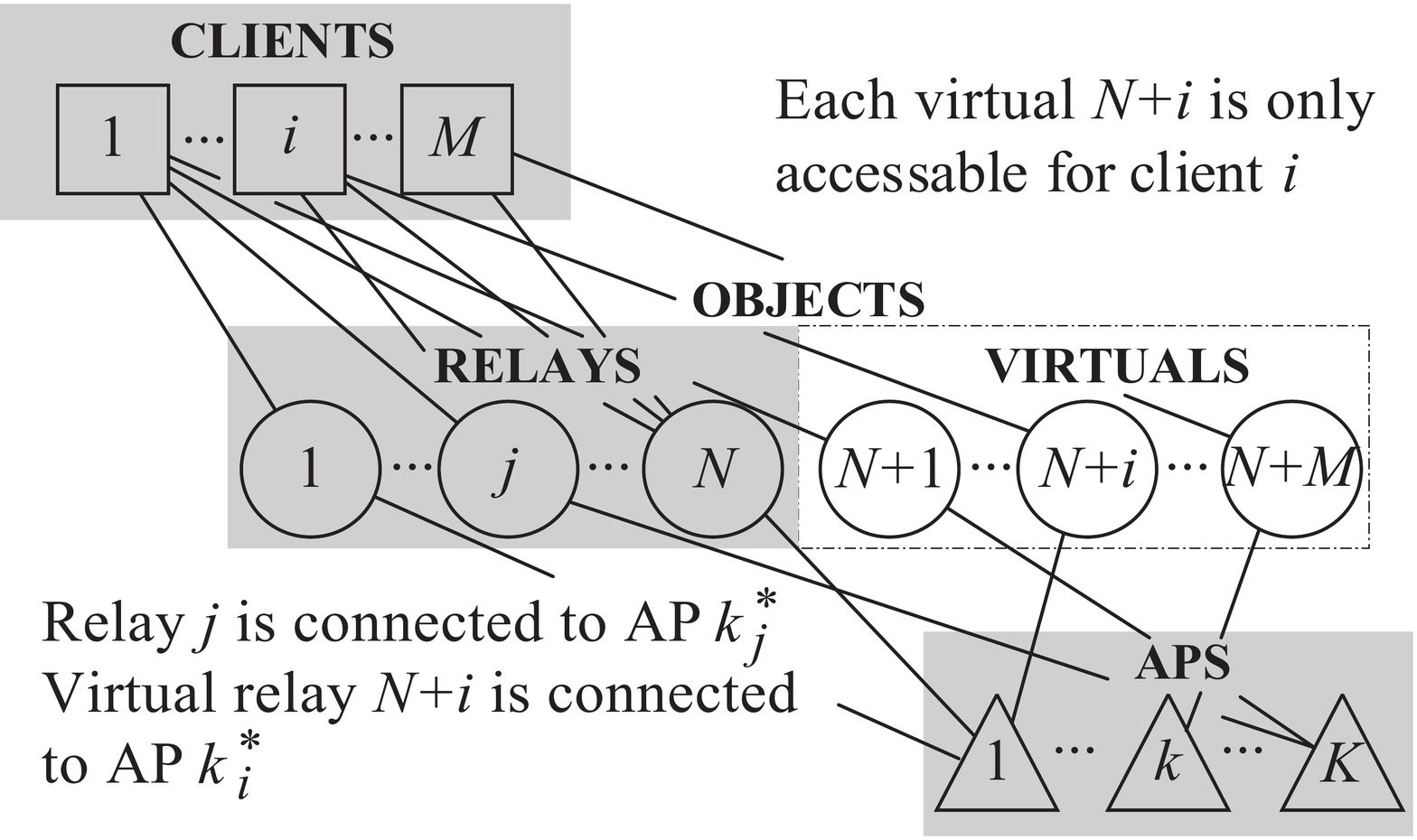}\vspace{-0.15in}
\caption{Illustration of the transformation for client-object set. Here $\mathcal N(i)=\{1,\dots,N\}$, and $\mathcal Q (i) = \{ 1, \dots, N, N+i\}$, where object $N+i$ represents the best AP for client $i$, $k_i^*$. }
\label{fig.example_transfer}\vspace{-0.2in}
\end{figure}

\begin{proposition}
Let $y^*_{(i,q)}$ be the optimal solution of problem~\eqref{asymmetric-assign-problem}. An optimal solution of problem~\eqref{multidim-assign-problem} is
\begin{align}
\label{eq.transfer-reconstruct}
    x_{(i,q_i^*,k_{q_i^*}^*)} &= \left\{ \begin{array}{ll} 1 & \textrm{if } y_{(i,q_i^*)} = 1 \textrm{ where } q_i^* \leq N \\
    0 & \textrm{otherwise}\end{array} \right.
    \\\label{eq.transfer-reconstruct1}
    x_{(i,k_i^*)} &= \left\{ \begin{array}{ll} 1 & \textrm{if } y_{(i,q_i^*)} = 1 \textrm{ where } q_i^* \geq N+1\\
    0 & \textrm{otherwise}\end{array} \right.
\end{align}
\end{proposition}
\begin{proof}
The existence of optimal binary solutions is assured by the theory for assignment-like problems by Proposition 5.6 in~\cite[\S~5]{Bertsekas1998}. 
Then, when we solve problem~\eqref{multidim-assign-problem} (assuring constraint~\eqref{load_balance_contraint}), we can associate every AP to more than one client or relay without reducing the throughput (transmission rate). Thus based on Proposition~\ref{prop.uniform_connection} not only constraints~\eqref{load_balance_contraint} is fulfilled, but also optimal assignments between AP and relays (clients) are achieved by using $\mathcal K^*(i)$, $\mathcal K^* (j)$ instead of $\mathcal K(i)$, $\mathcal K(j)$. Moreover, the optimal objective value achieved by problem~\eqref{asymmetric-assign-problem} equals to that by problem~\eqref{multidim-assign-problem}. Thus if an optimal assignment for problem~\eqref{multidim-assign-problem} exists, the optimal objective value can be achieved by both problems~\eqref{multidim-assign-problem} and~\eqref{asymmetric-assign-problem}. Finally, by expression~\eqref{eq.transfer-reconstruct} and \eqref{eq.transfer-reconstruct1}, the optimal assignment for problem~\eqref{asymmetric-assign-problem} gives an optimal assignment for problem~\eqref{multidim-assign-problem}, which completes the proof.
\end{proof}
Recall that there may be multiple optimal solutions for problem~\eqref{multidim-assign-problem}, if one or more clients can achieve same transmission rates to different APs via a relay by optimal assignments between clients and relays.

\subsection{Minimum Cost Flow Problem}\vspace{-0.05in}
\label{prop.min_cost_flow}
We now convert problem~\eqref{asymmetric-assign-problem} into a typical minimum cost flow problem following the methodology proposed in~\cite{Bertsekas1998}. We replace \emph{maximization} with \emph{minimization}, reversing in parallel the sign of $\beta_{(i,q)}$, and we introduce a virtual supersource client $s$ that is connected to each object $q$ through a zero cost artificial arc $(s,q)$ with feasible flow range $[0,\infty)$ for each object $q$. The supersource client $s$ generates traffic equal to $Q-M$ units while the supply for each remaining clients is of one unit. As a consequence one unit of traffic is the output of each object $q$. We refer the reader to~\cite[\S~7.2.2]{Bertsekas1998} for more details. The resulting problem is 
\begin{subequations}
\label{eq.min_cost}
\begin{align}
	\underset{y_{(i,q)}}{\textrm{minimize}} &\quad \sum_{(i,q)\in\mathcal A^*}  -\beta_{(i,q)} y_{(i,q)} \\
    {\rm s.t.}  & \quad \sum_{q\in\mathcal Q(i)} y_{(i,q)} = 1,\ \forall i \in \mathcal M \,,\\
                & \quad \sum_{i\in \mathcal M(q)} y_{(i,q)} + y_{(s,q)} = 1,\ \forall q \in \mathcal Q(i) \,, \\
                & \quad \sum_{q\in\mathcal Q} y_{(s,q)} = Q - M\,,\\
                & \quad y_{(i,q)}, y_{(s,q)} \geq 0,\ \forall (i,q)\in\mathcal A^* \,, q \in \mathcal Q\,,
\end{align}
\end{subequations}
where the decision variables $y_{(i,q)}$ were extended to include also $s$. By using the terminology of network optimization, $y_{(i,q)}$ has the meaning of amount of flow between $i$ and $q$. The first two constraints ensure that the flow \emph{supply} of each client $i$ is one unit, and a flow of one unit will reach every object $q$ respectively. The third constraint declares that $s$ is the supersource client and the flow that generates is of $Q-M$ units. The last two constraints declare that the flow of each arc may be infinite, where an arc between $i$ and $q$ denotes the connection $(i,j)$. A solution to the minimum cost flow problem \eqref{eq.min_cost} is the same to the initial~\eqref{asymmetric-assign-problem}~\cite[\S~7]{Bertsekas1998}.

By using the duality theory for minimum cost flow problems \cite[\emph{\S 4.2}]{Bertsekas1998} we formulate the dual problem 
\begin{subequations}
\label{dual_problem}
\begin{align}
    \underset{p_q, \pi_i, \lambda}{\textrm{minimize}}  & \quad \sum_{i\in \mathcal M} \pi_i + \sum_{q\in\mathcal Q} p_q - (Q-M)\lambda \\
    {\rm s.t.}      & \quad \pi_i+p_q \geq \beta_{(i,q)},\ \forall (i,q)\in \mathcal A^* \,,\\
                    & \quad \lambda \leq p_q,\ \forall q \in \mathcal Q\,.
\end{align}
\end{subequations}
where $-\pi_i$ is the Langrangian multiplier introduced to represent the price (or benefit due to the negative sign) of each client $i$, $p_q$ represents the price for object $q$ and $\lambda$ the price for $s$. The optimal solution to problem \eqref{dual_problem} allows us to derive the optimal solution to \eqref{asymmetric-assign-problem} \cite[\emph{\S 4.2, \S 5}]{Bertsekas1998}. 

In order to solve \eqref{dual_problem} we need some technical intermediate results. We start by giving the definition of $\epsilon$-\textit{Complementary Slackness} ($\epsilon$-CS)
\begin{definition}$\epsilon$-CS: 
Let $\epsilon$ be a positive scalar. An assignment $\mathcal S$ and a pair $(\pi,p)$ are said to satisfy $\epsilon$-CS if
\begin{align*}
    \pi_i + p_q \geq \beta_{(i,q)} - \epsilon\,, &\ \forall (i,q) \in \mathcal{A}^*\,,\\
    \pi_i + p_q = \beta_{(i,q)}\,, &\ \forall (i,q) \in \mathcal S\,, \\
     p_n \leq \min_{m\in \mathcal Q(\mathcal S)}p_m\,, &\ \forall n  \notin \mathcal Q (\mathcal{S})\,,
\end{align*}
where the set $\mathcal{Q}(\mathcal S) = \left\{ q | (i,q) \in \mathcal S,  \forall i \right\}$
contains all objects assigned under assignment $\mathcal S$.
\end{definition}

\begin{proposition}\label{prop:ecs}
Consider problems \eqref{asymmetric-assign-problem} and \eqref{dual_problem}. Let $\mathcal S$ be a feasible solution for problem~\eqref{asymmetric-assign-problem} and consider a dual variable pair ($\pi, p$). If $\epsilon$-CS conditions are satisfied by $\mathcal S$ and $\pi, p$, then $\mathcal S$ is optimal for problem \eqref{asymmetric-assign-problem}.
\end{proposition}
\begin{proof}
The proof naturally results from the proof of Proposition 7.7 in~\cite{Bertsekas1998}.
\end{proof}

Based on Proposition \ref{prop:ecs}, we are now in the position to present
the solution method to problem \eqref{dual_problem} by iterative centralized auction algorithms. This is the fundamental step that we need to develop a fully distributed solution mechanism. The auction algorithm is based on two phases: forward and reverse auction. 
%

We first apply forward auction. It starts from a feasible assignment $\mathcal S$ and the corresponding benefit-price pair, $(\pi, p)$, that satisfy the first two $\epsilon$-CS conditions. It picks client $i$, one of the unassigned clients under $\mathcal S$. Client $i$ finds the best object $q_i$ that provides most benefits among the objects in $\mathcal Q(i)$. Then, it bids for object $q_i$. Such an object updates its price while it is assigned to client $i$. If object $q_i$ was connected to other client $i'$, client $i'$ is left now unassigned. Forward auction terminates when all clients are assigned to objects. Since some of the objects (relay nodes or APs) are still unassigned after the application of forward auction, a reverse auction is applied. It gets as input the assignment achieved by the forward auction $\mathcal S$ and $(\pi,p)$, and checks whether there exist unassigned objects whose prices are larger than the minimum assigned object price $\lambda$ under previous assignment result, such that $\lambda = \min_{q\in\mathcal Q(\mathcal S)} p_q$.
If there are no such objects, algorithm terminates. Otherwise, it picks object $q$ whose price is larger than $\lambda$,
and finds the best client $i_q$ that provides highest price. Then, object $q$ decreases its price to attract client $i_q$. Client $i_q$ and object $q$ update their benefit and price respectively. Moreover, object $q$ connects to client $i_q$. The reverse auction algorithm terminates when all unassigned objects $q$ satisfy $p_q \leq \lambda$. Note that the scalar $\lambda$ is kept fixed throughout the algorithm and the last $\epsilon$-CS condition is satisfied upon termination. 

We refer the reader to \cite[\emph{\S 7}]{Bertsekas1998} for a detailed description of centralized auction algorithms. We are now in the position of establishing a distributed solution method in next section. 
\section{Distributed Auction Algorithms}\vspace{-0.05in}
\label{distr_algo}
In this section, we propose distributed algorithms for the solution of problem~\eqref{dual_problem} and therefore problem~\eqref{multidim-assign-problem}. First we focus on static networks, and then we consider dynamic networks where clients or relays can join or leave. 
\subsection{Static networks}\vspace{-0.05in}
The distributed solution method is based on the application of Algorithm~\ref{alg_d_c} by the clients, and Algorithm~\ref{alg_d_n} by the relays. In the solution algorithms, the vector $P_i \in \mathbb{R}^{N}$ denotes the prices vector for the relays (stored in client $i$), $p_j$ denotes the price of a relay node $j$ (stored in relay $j$), and $k_i^*$ represents AP $k_i^*$ for client $i$. In what follows we present the basic steps to establish the distributed algorithms. 

Initially the prices of all the objects are set to zero in both algorithms. On the client side (Algorithm~\ref{alg_d_c}), every client $i$ fulfilling conditions in Line 11 finds the best object $q_i$ using the local knowledge of the prices in Line 12. From Line 13 to 18, client $i$ calculates the largest bid for the object $q_i$. Then in Line 19, it sends the request to the object $q_i$. On the relay (object) side (Algorithm~\ref{alg_d_n}), when object $q_i$ receives the request from clients with different bids, as described in Line 2, it chooses the best client $i_q$ that provides highest bid and higher price compared to the old price $p_j$ (Line 3 to 4). Then, object $q_i$ updates its price and feedbacks the latest price to all requested clients as described in Line 6 to Line 9. The auction algorithms terminate when there is no client fulfilling conditions in Algorithm~\ref{alg_d_c}, Line 9.

\begin{algorithm}[t]\scriptsize
\caption{\small Distributed Auction Algorithm for Client $i$}
\label{alg_d_c}
\begin{algorithmic}[1]
\REQUIRE Initialize $q_i=k_i^*, P_i=\mathbf{0}$
\WHILE {\TRUE}
\IF {receive {\bf no} and new price $p_{q_i}$ from previous object $q_i$}
\STATE Connect to object $k_i^*$
\STATE Let $[P_i]_{q_i} = p_{q_i}$ and $q_i = k_i^*$
\ENDIF
\IF {previous object $q_i$ leaves}
\STATE Connect to object $k_i^*$ and let $q_i = k_i^*$
\ENDIF
\IF{$q_i = k_i^* \neq \arg \max_{q\in \mathcal Q(i)} \left\{ \beta_{(i,q)} - [P_i]_q \right\}$ }
\STATE Find the bid $b_{i{q_i}}$ for the best object $q_i$, such as:\\
$q_i = \arg \max_{q\in \mathcal Q(i)} \left\{ \beta_{(i,q)} - [P_i]_q \right\},$
\STATE Find $u_i = \max_{q\in \mathcal Q(i)} \left\{\beta_{(i,q)} - [P_i]_q\right\},$
\STATE Find $\omega_i = \max_{q\in \mathcal Q(i), q\neq q_i} \left\{ \beta_{(i,q)}- [P_i]_q \right\},$
\IF {$q_i$ is the only client in $ \mathcal Q(i)$}
\STATE Let $\omega_i \rightarrow -\infty$
\ENDIF
\STATE Find the bid $b_{i{q_i}} = p_{q_i}+u_i- \omega_i+\epsilon$
\STATE Send the transmit request along with the bid $b_{i{q_i}}$ to object $q_i$
\STATE Receive the respond message, (\textbf{yes} or \textbf{no}) and $p'_{q_i}$ from object $q_i$
\IF{respond is \textbf{yes}}
\STATE Connected to object $q_i$
\ELSE
\STATE Connected to the AP in set $\mathcal Q(i)$ and Let $q_i = k_i^*$
\ENDIF
\STATE Let $p_{q_i}=p'_{q_i}$
\ENDIF
\ENDWHILE
\end{algorithmic}
\end{algorithm}

\begin{proposition}
\label{prop.finity-convergence}
Consider $M$ clients, $N$ relay nodes, and $K$ APs. The distributed auction algorithms given in Algorithm~\ref{alg_d_c} and Algorithm~\ref{alg_d_n} terminate within a finite number of iterations bounded by $M N^2 \lceil \Delta /\epsilon \rceil$, where $\Delta = \max_{(i,q)\in \mathcal A^*} \beta_{(i,q)} - \min_{(i,q)\in \mathcal A^*}\beta_{(i,q)}$.
\end{proposition}

\begin{proof}
Firstly notice that client $i$ can at least connect to the AP $k_i^*$. Thus, problem~\eqref{asymmetric-assign-problem} has a feasible assignments with positive net benefit (total throughput). Moreover, in every iteration there exist clients that are not yet informed for the latest price of some objects. This implies that these clients may place low bids for expensive relay nodes (objects). However, they will be informed after the biding. Based on this observation, we can disregard all lower bids and only consider bids by informed clients that increase the actual price of the relay node. Hence, we only need to show that every relay node can only receive a finite number of such bids.

Note that whenever $m$ bids are placed for a relay node, its price must increase by at least $m\epsilon$. Thus, when $m$ is sufficiently large, the relay node will become too expensive to be attractive compared to other relay nodes that have not yet received any bids. It follows that there is a limited number of bids that any relay node can receive by informed clients. Therefore, the auction will continue until each one of the clients has been associated to one object.

In the worst case, we consider that all the clients persistently place minimum bid increments $\epsilon$. Furthermore, they won't win the object until the local price vectors in clients is updated to the latest. Without considering the price update, the number of iterations of the auction algorithm is bound by $N \lceil \Delta/\epsilon \rceil$, because every node $i$ will eventually be associated to one object (the best AP) when the benefit of all the relay nodes in $\mathcal Q(i)$ is lower than that of the best AP. Meanwhile, in every iteration, the throughput benefit decreases monotonically at least by $\epsilon$. On the other side, the number of iterations for price update is bounded by $MN$. Thus we can find that the number of iterations of the distributed algorithms is bounded by $M N^2 \lceil \Delta/\epsilon \rceil$, which completes the proof.
\end{proof}

\begin{algorithm}[t]\scriptsize
\caption{\small Distributed Auction Algorithm for Relay $j$}
\label{alg_d_n}
\begin{algorithmic}[1]
\REQUIRE Initialize the client $i_j=0$, and price $p_j=0$
\WHILE {\TRUE}
\IF {there are transmit request from clients $i\in \mathcal M(j)$ with bids $b_i$}
\STATE Find the best client $i_j$ such that:\\
$i_j = \arg \max_{i\in \mathcal M(j)} \left\{ b_i \right\},$
\STATE Find $p_{i_q} = \max_{i\in \mathcal M(j)} \left\{ b_i  \right\},$
\IF{$p_{i_q}-p \geq \epsilon$}
\STATE Set $p_j=p_{i_q}$
\STATE Send respond message, \textbf{no} and $p_j$, to the other clients that requested and previous client $i_j$
\STATE Send respond message, \textbf{yes} and $p_j$, to client $i_q$
\STATE Set $i_j = i_q$
\ELSE
\STATE Send respond message, \textbf{no} and $p_j$, to all clients that requested
\ENDIF
\ENDIF
\ENDWHILE
\end{algorithmic}
\end{algorithm}

\begin{remark}\vspace{-0.1in}
Note that the bound on the number of iterations is conservative and it is based on the absence of broadcast transmission in the network. Otherwise, if every relay node $j$ can broadcast its latest prices to the clients in set $\mathcal M(j)$, the iteration would be bounded by $N^2 \lceil \Delta/\epsilon \rceil$.
\end{remark}

\begin{proposition}\vspace{-0.1in}
\label{prop convergence_distributed}
Let $\epsilon$ be a desired positive constant. The final assignment obtained by Algorithm~\ref{alg_d_c} and Algorithm~\ref{alg_d_n} is within $M\epsilon$ of the optimal assignment benefit of problem~\eqref{asymmetric-assign-problem}. The final assignment is optimal if $\beta_{(i,q)}$, $\forall (i,q)\in\mathcal A^*$, is integer\footnote{If the benefits are rational numbers, they can be scaled up to integer
by multiplication with a suitable common number.} and $\epsilon < 1/M$.
\end{proposition}
\begin{proof}
The total number of objects is $M+N$, and recall that $\epsilon>0$. Then, given any assignment $\{\mathcal S|(i,q_i),i=1,\dots,M\}$, the net benefit satisfies
\begin{equation*}
    \sum_{i=1}^{M} \beta_{(i,q_i)} \leq \sum_{q=1}^{M+N} p_q + \sum_{i=1}^{M} \max_q \{ \beta_{(i,q)} - p_q \}\,,
\end{equation*}
for any set of prices, since the second term of the right hand side is no less than $\sum_{i=1}^{M} (\beta_{(i,q_i)}-p_{q_i})$,
while the first term is no less than $\sum_{i=1}^{M} p_{q_i}$. Therefore, $B^* \leq D^*$, 
where $B^*$ is the optimal total assignment benefit for problem~\eqref{multidim-assign-problem}
\[
    B^* = \max_{\substack{q_i\in \mathcal Q(i), \\q_i \neq q_m, \textrm{if }i\neq m}} \sum_{i=1}^{M} \beta_{(i,q_i)}\,,
\]
and $D^*$ is the optimal minimum for dual problem~\eqref{dual_problem}, denoted by
\[
    D^* = \min_{\substack{p_q\\q=1,\dots,M+N}} \left\{ \sum_{q=1}^{M+N} p_q + \sum_{i=1}^{M} \max_{q\in \mathcal Q(i)} \{ \beta_{(i,q)} - p_q \} \right\}\,.
\]
On the other hand, consider the assignment $\mathcal S$ together with the set of price $(\bar{p}_q)_{ q = 1,\dots,M+N}$ that is stored in the relays (objects). Moreover, notice that once a bid by client $i$ is accepted by object $q_i$, we have
\[
    \beta_{(i,q_i)} - \max_{q\in \mathcal Q(i), q\neq q_i} \{\beta_{(i,q)}- [P_i]_q\} \geq \bar{p}_{q_i} \,,
\]
where $[P_i]_q$ is the local price in client $i$ for relay $q$. The previous expression indicates that
\begin{align*}
    \beta_{(i,q_i)} - \bar{p}_{q_i} 	\geq& \max_{q\in \mathcal Q(i), q\neq q_i} \{\beta_{(i,q)}-[P_i]_q\} \\
    								\geq& \max_{q\in \mathcal Q(i), q\neq q_i} \{\beta_{(i,q)}-\bar{p}_q\}  \,,
\end{align*}
since $[P_i]_q \leq \bar{p}_q$ for every object $q$. Thus object $q_i$ is the best object for client $i$, even though $[P_i]_{q_i} \leq \bar{p}_{q_i}$. Furthermore, we have
\[
    \beta_{(i,q_i)} - \bar{p}_{q_i} \geq \max_{q\in \mathcal Q(i)} \{\beta_{(i,q)} - \bar{p}_q \} - \epsilon\,.
\]
Now considering all $i$s, we see that
\begin{align*}
    \sum_{i=1}^{M} \bigg( \bar{p}_{q_i} + & \max_{q\in \mathcal Q(i)}\{ \beta_{(i,q)}-\bar{p}_q\} \bigg) =\\
     & \sum_{q=1}^{M+N} \bar{p}_q  + \sum_{i=1}^{M} \left( \max_{q\in \mathcal Q(i)}\{ \beta_{(i,q)}-\bar{p}_q\} \right) \geq D^* \,,
\end{align*}
since the prices of unassigned object are equal to zero. Thus
\begin{align*}
    D^* \leq &  \sum_{i=1}^{M} \left( \bar{p}_{q_i} + \max_{q\in \mathcal Q(i)}\{ \beta_{(i,q)}-\bar{p}_q\} \right) \\
    \leq& \sum_{i=1}^{M} \beta_{(i,q_i)} + M\epsilon \leq B^*+M \epsilon\,.
\end{align*}
However, we showed earlier that $B^*\leq D^*$, and it follows that the net benefit $\sum_{i=1}^{M}\beta_{(i,q)}$ is within $M \epsilon$ of the optimal value $B^*$. Now consider $\epsilon < 1/M$, then the assignment obtained is strictly within 1 of being optimal. Furthermore, if all parameters $\beta_{(i,q)}$ are integer, the minimum difference from optimal is 1, which forces the assignment be optimal. This completes the proof. 
\end{proof}

\begin{algorithm}[t]\scriptsize
\caption{\small Distributed Reverse Auction Algorithm for Relay $j$}
\label{alg_d_r_n}
\begin{algorithmic}[1]
\REQUIRE Initialize the client $i_j=0$ and price $p_j=0$
\WHILE {\TRUE}
\IF {join a network}
\STATE Set $p_j = \infty$
\ENDIF

\IF {there are transmit request from clients $i\in \mathcal M(j)$ with information $b_i$}
\STATE \dots \COMMENT{Algorithm~\ref{alg_d_n} Line 3 $\sim$ 12}
\ENDIF

\IF {unassigned \AND $p_j \neq 0$}
\STATE Send survey request to clients in $\mathcal{M} (j)$
\STATE Receive $(\beta_{(i,j)}, \beta_{(i,q_i)},[P_i]_{q_i})$ from clients $i$ in $\mathcal{M}'(j)$,\\
 where $\mathcal{M}'(j) = \{i| \beta_{(i,j)} - (\beta_{(i,q)}-[P_i]_{q_i})\geq \epsilon, \forall i\in\mathcal M\}$
\STATE Find the best client $i_q$ such that\\
$i_q = \arg \max_{i\in \mathcal M'(j)} \left\{ \beta_{(i,j)} - (\beta_{(i,q)}-[P_i]_{q_i}) \right\},$
\STATE Find $\gamma_j = \max_{i\in \mathcal M'(j)} \left\{ \beta_{(i,j)} - (\beta_{(i,q)}-[P_i]_{q_i}) \right\},$
\STATE Find $\omega_j = \max_{i\in \mathcal M'(j),i\neq i_q} \left\{ \beta_{(i,j)} - (\beta_{(i,q)}-[P_i]_{q_i}) \right\},$
\STATE Let $\lambda = 0$, since $M\geq N$
\IF {$i_q$ is the only client}
\STATE Let $\omega_j \rightarrow -\infty$
\ENDIF
\IF {$\lambda \geq \gamma_j - \epsilon$}
\STATE Let $p_j = 0$
\ELSE
\STATE Connect to client $i_q$
\STATE Let $p_j = \gamma_j - \min\left\{ \gamma_j - \lambda, \gamma_j - \omega_j + \epsilon  \right\}$, $i_j = i_q$
\ENDIF
\ENDIF
\ENDWHILE
\end{algorithmic}
\end{algorithm}

\subsection{Dynamic networks}\vspace{-0.05in}

Suppose now that some clients or relays leave or join the mmW network. Obviously, the optimal assignment $\mathcal S^*$ will change. To handle this dynamic evolution of the network and to improve the stability of our approach, we propose a distributed reverse auction algorithm. 

A client joining the network can be seen as a relay leaving the network, since the network gains one more unassigned client. Similarly, a relay joining is the same as a client leaving the network, since there will be one more unassigned relay in the assignment. In case that there exists one more unassigned client, Algorithm~\ref{alg_d_c} and Algorithm~\ref{alg_d_n} can still work well if the new unassigned client is initialized based on the requirements of Algorithm~\ref{alg_d_c}. The unassigned client is always able to place a bid to its best relay and get its local price vector update. On the other hand, an unassigned relay with a high price is not attractive to clients any more, which keeps the network away from the optimal assignment. The distributed reverse auction algorithm for the relays is described by Algorithm~\ref{alg_d_r_n}. The main idea is similar to the reverse auction algorithm presented in \cite[\emph{\S 7}]{Bertsekas1998}, in which every object tries to reduce its price to attract the best client. In particular, if a new relay $j$ joins a network, it initializes its price to $p_j = \infty$ as in Lines 2 to 4 in Algorithm~\ref{alg_d_r_n}. Moreover, if the relay is not connected to any client and its price is higher than $\lambda = 0$, it will invite the clients in the set $\mathcal M'(j)$, to potentially increase the benefit, and send a request message $(\beta_{(i,j)},\beta_{(i,q_i)},[P_i]_{q_i})$ as in Lines 9 to 10. Then relay $j$ finds the best client to be connected as in Lines 11 to 24.

\begin{table}[t]\vspace{-0.05in}
  \centering
  \scriptsize
  \caption{Simulation Parameters}\label{tab_sim_para}
   \vspace{-0.15in}
  \begin{tabular}{ccc}
    \hline
    \textbf{Parameters} & \textbf{Symbol}   & \textbf{Value} \\
    \hline
    System Bandwidth    & $W$               & 1200 MHz \\
    Transmission Power  & $P_T$             & 0.1 mW \\
    Background Noise    & $N_0$             & -134 dBm/MHz \\
    Path Loss Exponent  & $\eta$            & 2 \\
    Reference distance  & $d_0$             & 1 m \\
    Wave Length         & $\lambda$         & 5 mm \\
    Antenna Gains       & $G_R$ and $G_T$   & 1 \\
    \hline
  \end{tabular} 
\vspace{-0.1in}
\end{table}

\begin{proposition}
\label{prop.within-M-epsilon}
Consider problem~\eqref{asymmetric-assign-problem}. The distributed reverse auction algorithm described by Algorithm~\ref{alg_d_r_n} terminates with an assignment that is within $M\epsilon$ of being optimal.
\end{proposition}
\begin{proof}
The proof naturally results from the proof of Proposition 7.8 in~\cite{Bertsekas1998} with $\lambda = 0$.
\end{proof}

\section{Numerical Examples}\vspace{-0.05in}
\label{sec.numerical-examples}
In this section we compare our approach to \begin{inparaenum}[\itshape a\upshape)]
\item random association policy, \item RSSI-based policy defined by 802.11, and \item optimal solution to the optimization problem~\eqref{multidim-assign-problem} using centralized binary integer programming solver \texttt{bintprog} in Matlab\footnote{There are not approaches in literature to consider jointly the association and the relaying problems in mmW networks and therefore, it is not fair to provide any comparison to other approaches for traditional wireless networks.}.
\end{inparaenum}

We define the SNR operating point at a distance $d$ [distance units] form any AP as\vspace{-0.05in}
\begin{equation}\vspace{-0.05in} \label{eq:SNR}\nonumber
\texttt{SNR}(d) = \left\{ \begin{array}{ll}
  \displaystyle{{P}_{0}}\lambda^2/(16\pi^2N_0W) & d\leq d_0\\
 \displaystyle{{P}_{0}}\lambda^2/(16\pi^2N_0W)\cdot\left({d}/{d_0}\right)^{-\eta} &  \textrm{otherwise}\ .
   \end{array} \right.
\end{equation}
Circular cells as illustrated in Fig.~\ref{fig network} are considered, where the radius of each cell $r$ is chosen such that $\texttt{SNR}(r)=10$~dB. The APs are located such that the the distance between any consecutive APs is $1.1r$. The clients and relays are uniformly distributed at random over this area. The main parameters used in simulations are listed in Table~\ref{tab_sim_para}. To measure the average performance of the algorithms, we run several Monte-Carlo simulations considering $T=500$ time slots for each experiment, and $N_r = 1000$ experiments with different topologies. Then, the average results over all the experiments are plotted. We used typical values for $\epsilon$ ($\epsilon=0.1$) where not stated in the description, and we set $I=0$, as discussed in section \ref{sec.system-model}. Obstacles may block the LOS link between clients, relays, and APs with probability $0.1$ and for a duration of $10$ ms.

\begin{figure}[t]\vspace{-0.2in}
\centering 
	\subfloat[]{\hspace{-0.1in}
	\centering
	\includegraphics[width = 0.25\textwidth]{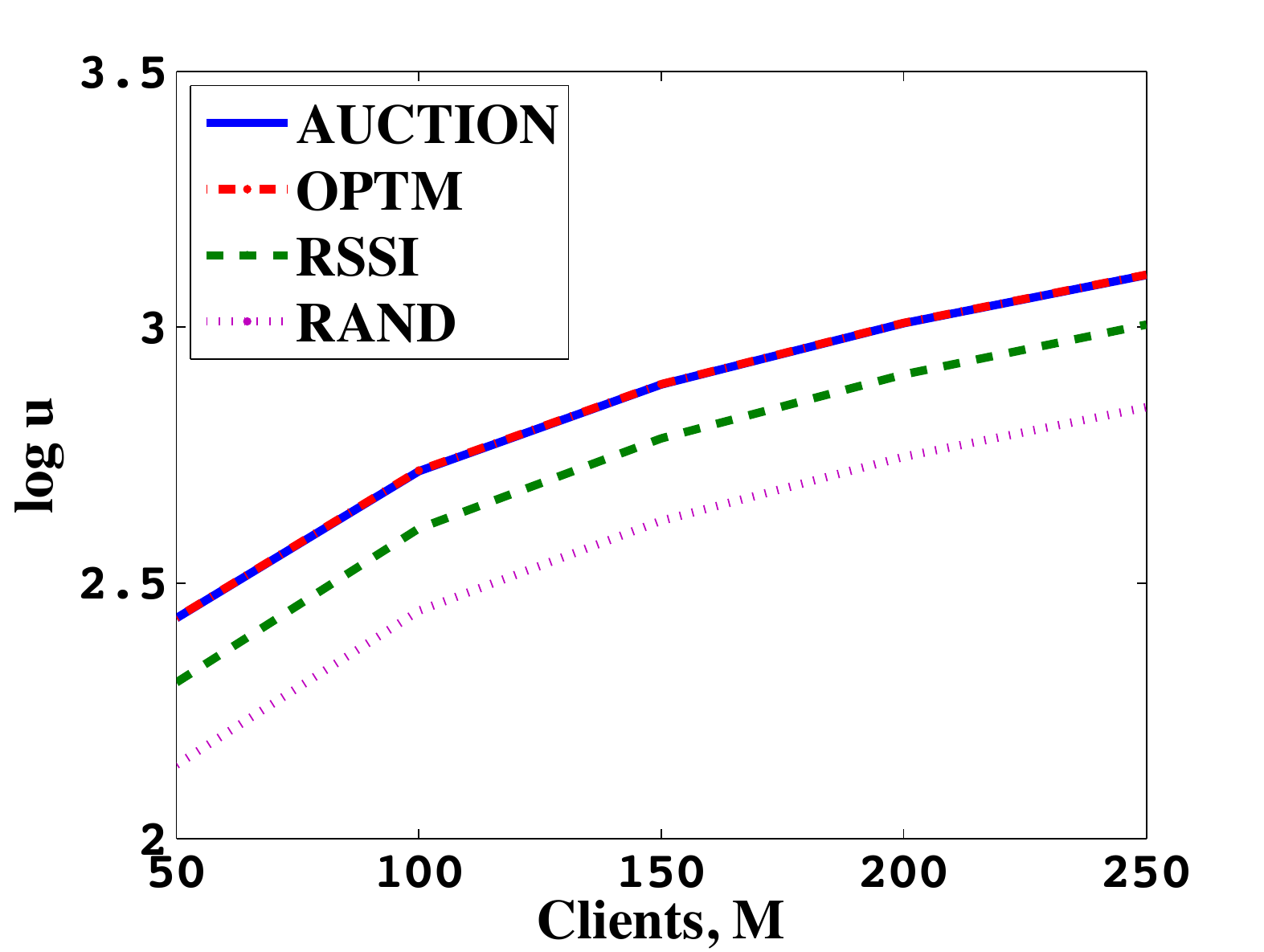}
	\label{fig.PCC_5AP}
	}
	\subfloat[]{
	\centering
	\includegraphics[width = 0.25\textwidth]{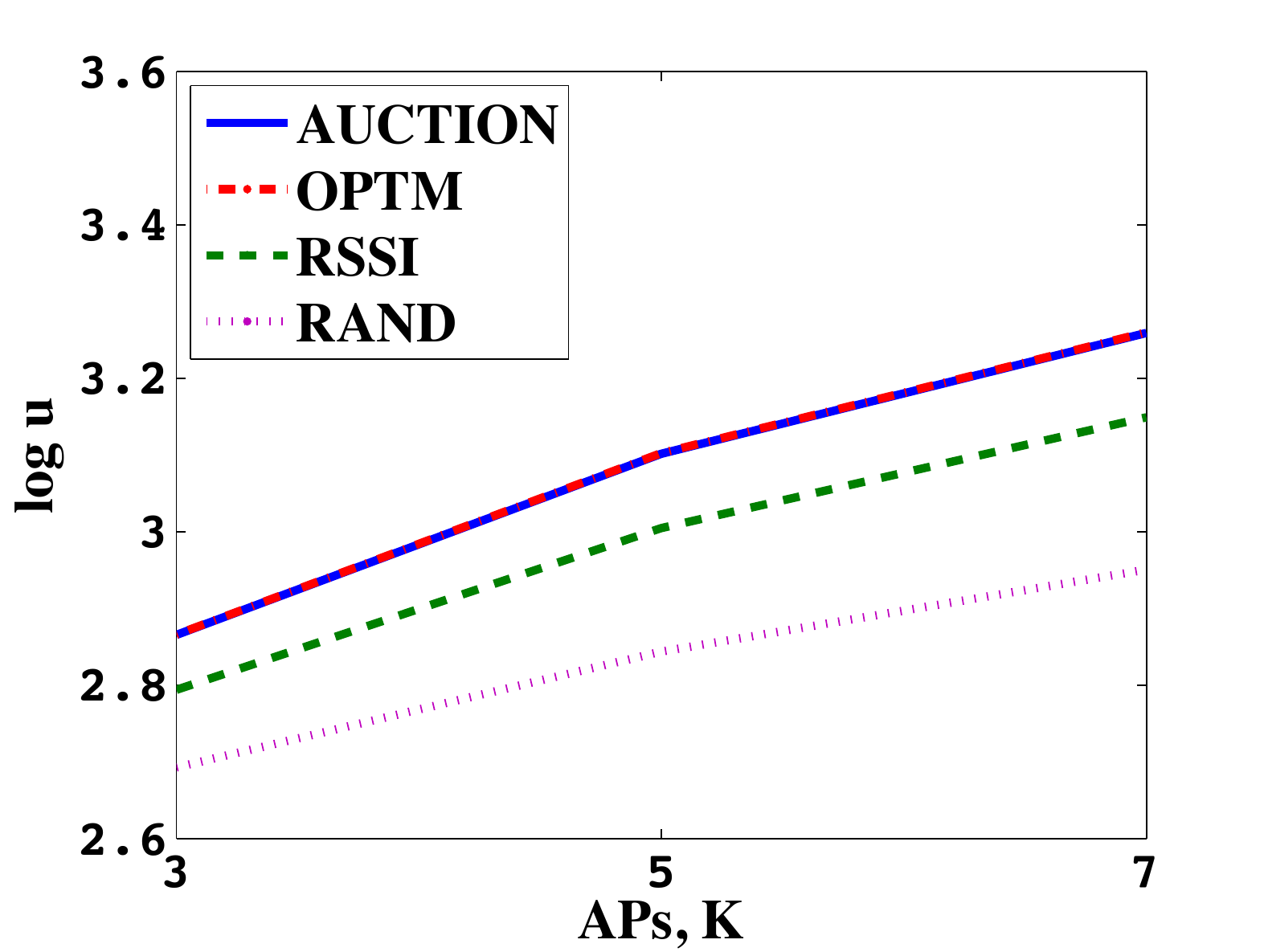}
	\label{fig.PCR_5AP}
	}
	\vspace{-0.1in}\caption{Objective values of AUCTION, OPTM, RAND, and RSSI (a) $\log u$ vs. number of clients with 10 APs and 25 relays; (b) $\log u$ vs. number of APs and 25 relays.}
	\label{fig.PCR}\vspace{-0.2in}
\end{figure}

\begin{figure}[t]\vspace{-0.05in}
  \centering
  \subfloat[]{\hspace{-0.1in}
    \centering
    \includegraphics[width = 0.25\textwidth]{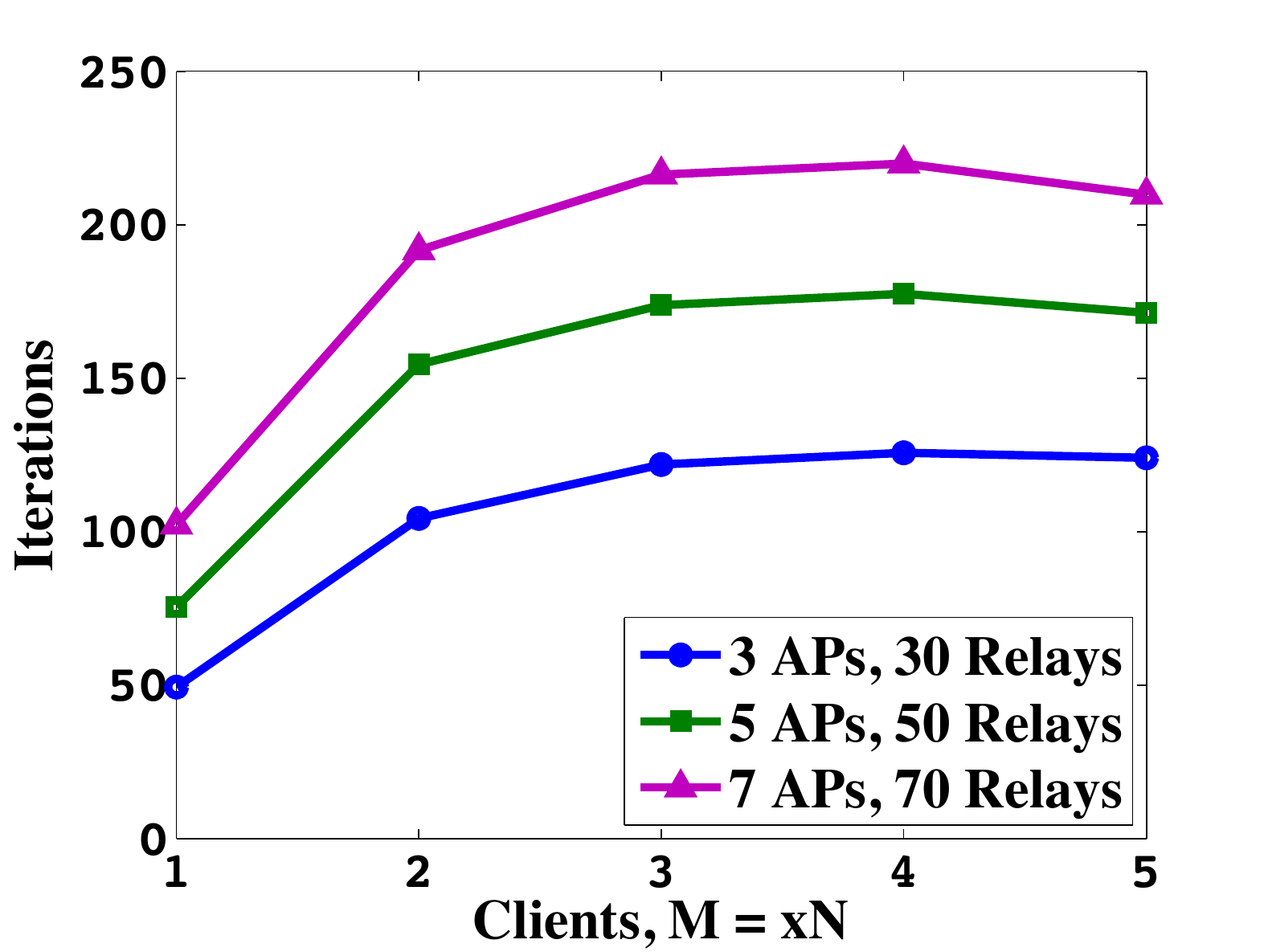}
    \label{fig.CSC}
    }
  \subfloat[]{
    \centering
    \includegraphics[width = 0.25\textwidth]{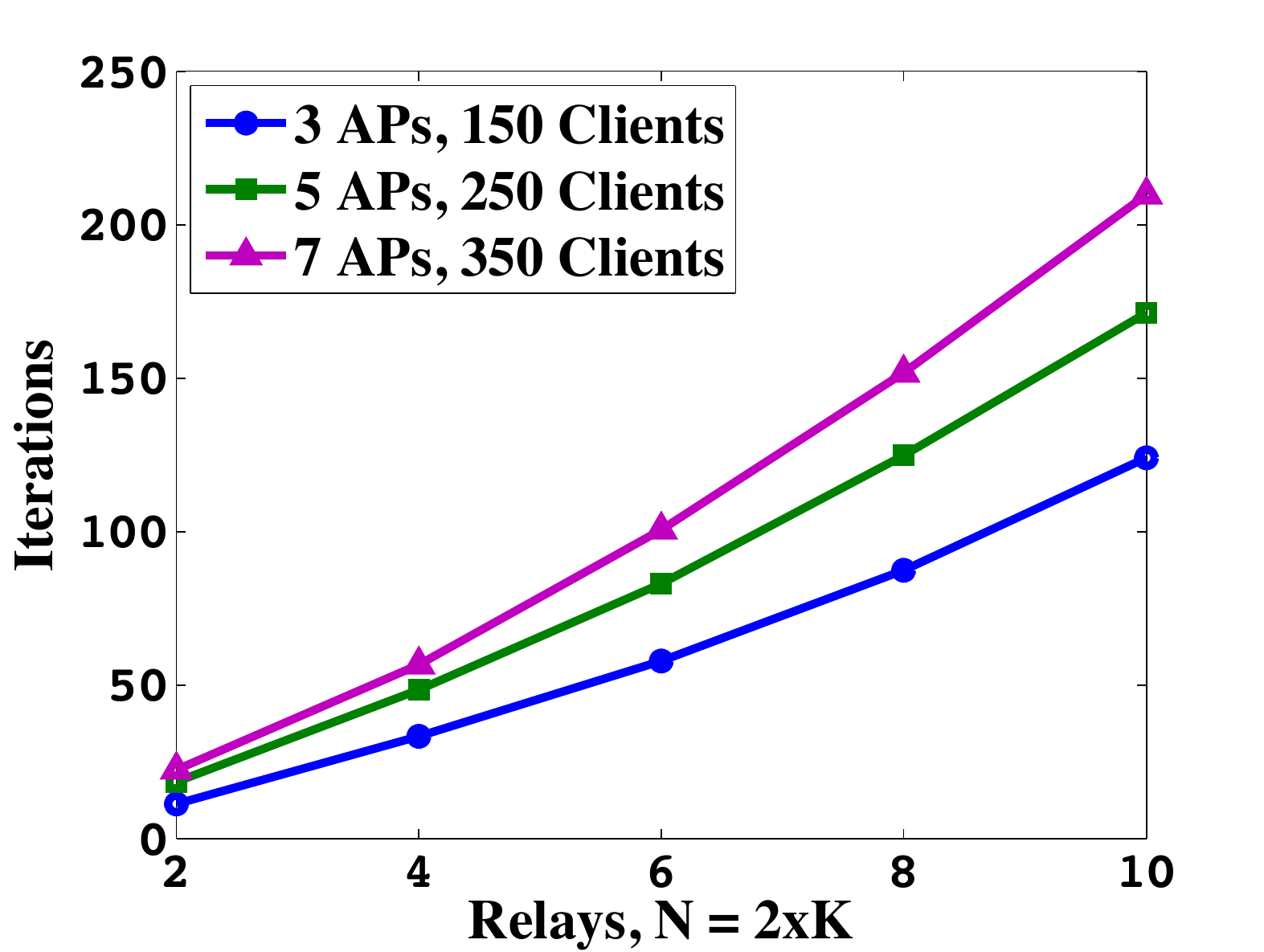}
    \label{fig.CSR}
    }
  \vspace{-0.1in}\caption{Convergence speed of the combined Algorithms~\ref{alg_d_c},~\ref{alg_d_n}, and~\ref{alg_d_r_n} as a function of the number of clients and relays. (a) Average number of iterations vs. number of clients, which varies from 1 to 5 times (x) the number of relays; (b) Average number of iterations vs. number of relays, which varies from 2 to 10 times (x) the number of APs.}
  \label{fig.CS}\vspace{-0.15in}
\end{figure}

We have implemented the distributed auction algorithms. Algorithm~\ref{alg_d_c} is executed by every client, while Algorithms~\ref{alg_d_n},~\ref{alg_d_r_n} are executed by every relay. Fig.~\ref{fig.PCR} shows the average objective value (average total network throughput across $N_r$ experiments) of problem \eqref{multidim-assign-problem}, obtained by the combined Algorithms~\ref{alg_d_c},~\ref{alg_d_n}, and~\ref{alg_d_r_n} after termination (AUCTION), in comparison to random association policy (RAND), RSSI-based policy (RSSI), and the optimal policy (OPTM). Results indicate that as the number of the supported clients (Fig.~\ref{fig.PCC_5AP}) and the APs (Fig.~\ref{fig.PCR_5AP}) varies, the distributed auction algorithms perform very close to the optimal policy. In Fig.~\ref{fig.PCR_5AP}, the clients density or the number of clients per AP is kept constant (30 clients per AP). This is consistent with Proposition~\ref{prop convergence_distributed}. RSSI-based, and we conclude that random policies give worse performance.

Next we observe the average termination behavior of the distributed auction algorithms considering different network sizes. The size is determined by the number of the APs, or the number of cells. The clients density or the number of clients per AP is also kept constant here. Fig.~\ref{fig.CS} shows the average number of iterations of our distributed algorithms while the number of clients (Fig.~\ref{fig.CSC}) and relays (Fig.~\ref{fig.CSR}) varies. Results show that there is a noticeable effect of varying number of clients ($M$) and the relays ($N$) on the termination time, thus they confirm Proposition~\ref{prop.finity-convergence}. The algorithms are faster for smaller $M$ and $N$ values.  

Fig.~\ref{fig.CSE} shows the convergence behavior of the distributed auction algorithms for different network sizes, while varying the values of $\epsilon$ used during the iterations. There is noticeable effect of $\epsilon$ on the termination time. The algorithms are faster for larger $\epsilon$, which agrees with Proposition~\ref{prop.finity-convergence}. On the other hand, Fig.~\ref{fig.PCE} indicates that the maximum distance from the optimal value of the objective function, $\Delta_{\max}$, rises while $\epsilon$ increases. This is inline with our analytical study and Proposition~\ref{prop convergence_distributed}, which indicate that the smaller the $\epsilon$, the closer to the optimal objective value the auction algorithms reach. Last but not least, $\Delta_{\max}$ is bounded by $\epsilon$, which is consistent with Proposition~\ref{prop.within-M-epsilon}.

\begin{figure}[t]\vspace{-0.2in}
\centering
   \subfloat[]{\hspace{-0.1in}
    \centering
    \includegraphics[width = 0.25\textwidth]{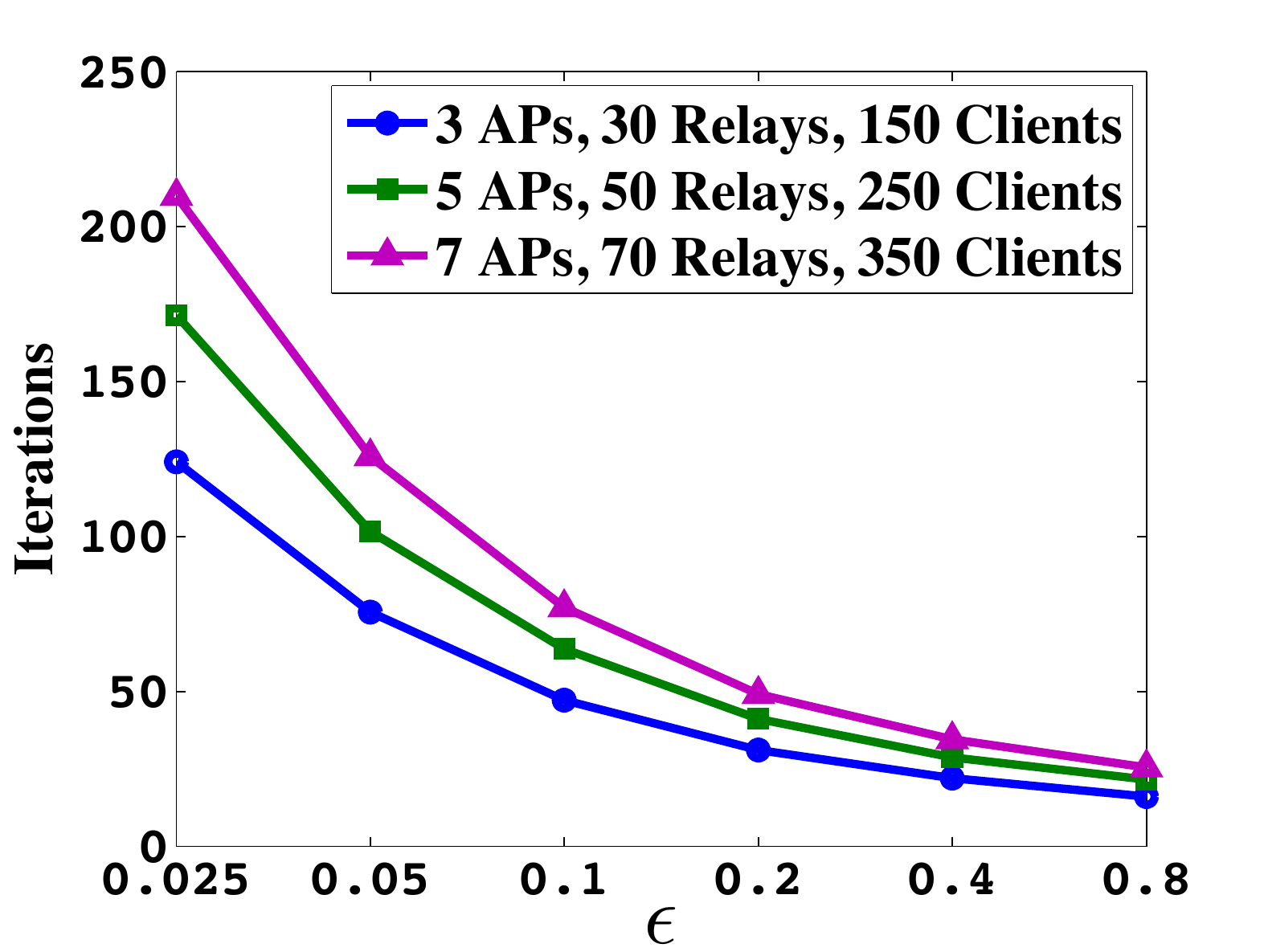}
    \label{fig.CSE}
    }
    \subfloat[]{
    \centering
  \includegraphics[width = 0.25\textwidth]{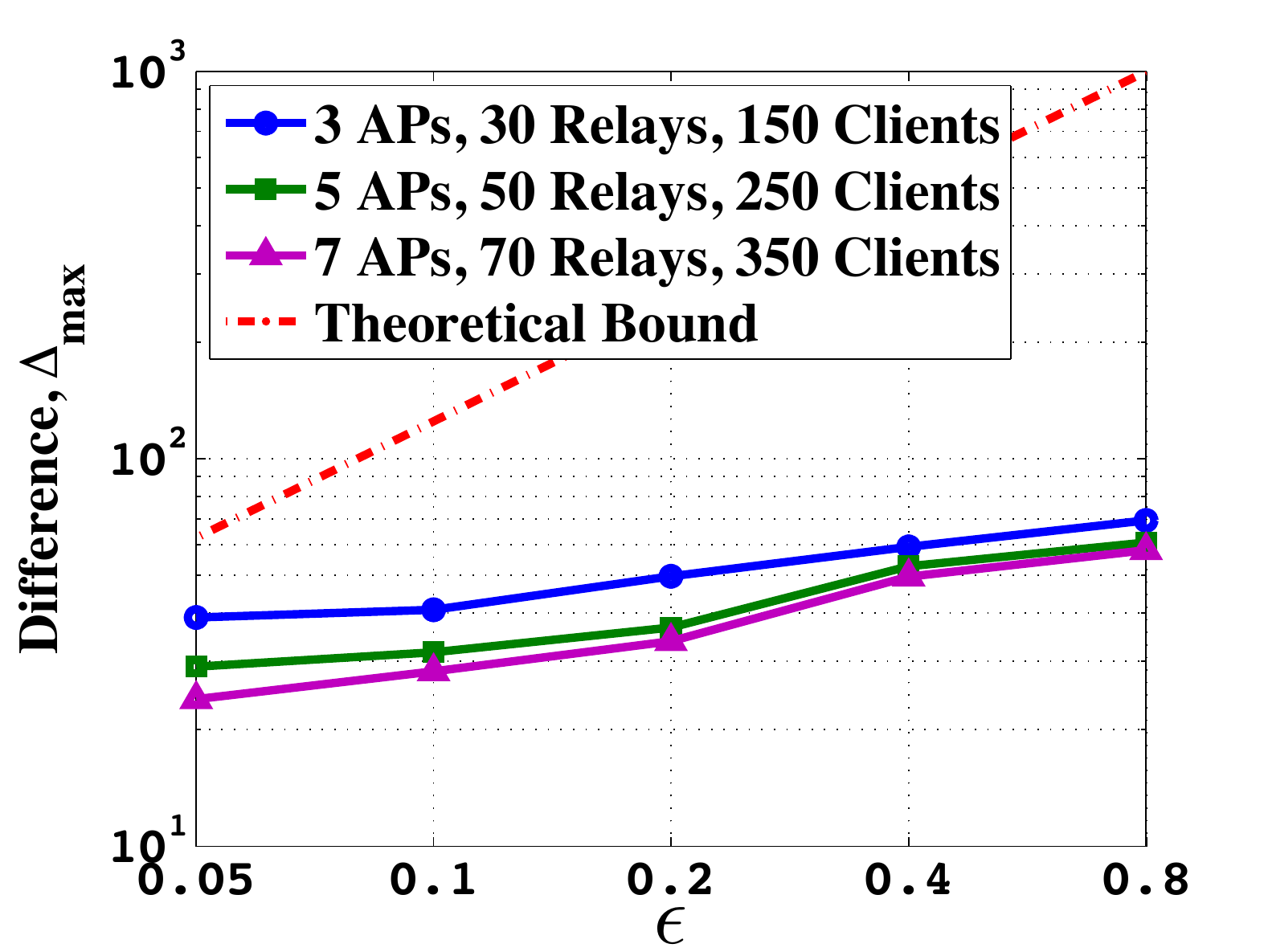}
  \label{fig.PCE}
  }
  \vspace{-0.1in}\caption{Convergence and maximum distance from the optimal objective value of problem \eqref{multidim-assign-problem}, $u$, when $\epsilon$ varies. (a) Number of iterations vs. $\epsilon$; (b) $\Delta_{\max}$ vs. $\epsilon$.}
  \label{fig.PC}\vspace{-0.15in}
  \end{figure}
  
Fig.~\ref{fig.DYN} depicts the performance of the auction algorithms in the dynamic case where 10 new clients and 5 new  relays join the network, at time slots 220 and 400 respectively. It is observed that the combined dynamic auction Algorithms \ref{alg_d_c} and  \ref{alg_d_r_n} converge faster, than  Algorithms \ref{alg_d_c} and  \ref{alg_d_n}, to average objective values close to optimal, ensuring the scalability and the stability of the proposed association and relaying approach. The results in~\ref{fig.DYN} agree with Proposition~\ref{prop.within-M-epsilon}.

Finally, we provide a statistical description of the speed of the auction algorithms in comparison to the centralized branch-and-bound algorithm that is applied by \texttt{bintprog} in Matlab. Therefore, we consider empirical cumulative distribution function (CDF) plots for the two approaches. Specifically, for each time slot, we store the total CPU time required for the auction algorithms, Algorithm 1 and 2, to terminate. We also use the total CPU time required by \texttt{bintprog} to find the \emph{optimal} solution to problem~\eqref{multidim-assign-problem}.  Results in Fig.~\ref{fig.TCC_CDF2C} show that auction algorithms are much quicker especially in cases with high load (large number of clients), where branch-and-bound algorithm suffers from high complexity. Fig.~\ref{fig.TCR_5AP250Client} depicts the average time required by these two approaches to find the optimal solution. Our auction algorithms outperform branch-and-bound algorithm as we vary the number of the supported relays in the network. 

\section{Conclusions}\vspace{-0.05in}
\label{sec.conclusion}
In this paper, we considered the problem of optimizing the allocation of the clients to APs and relays in mmW wireless access networks. The objective was to \emph{maximize the total throughput} that the clients get in the network. The resulting multi-assignment problem is combinatorial and non-convex. Thus, new distributed action algorithms were proposed to solve the problem considering both static and dynamic mmW networks. The performance of the proposed algorithms was studied and verified in comparison to standard approaches through theoretical and numerical analysis. Our results indicate that the proposed solutions could be well applied in the forthcoming mmW wireless access networks.

 \begin{figure}[t]\vspace{-0.15in}
\centering
  \includegraphics[width = 0.25\textwidth]{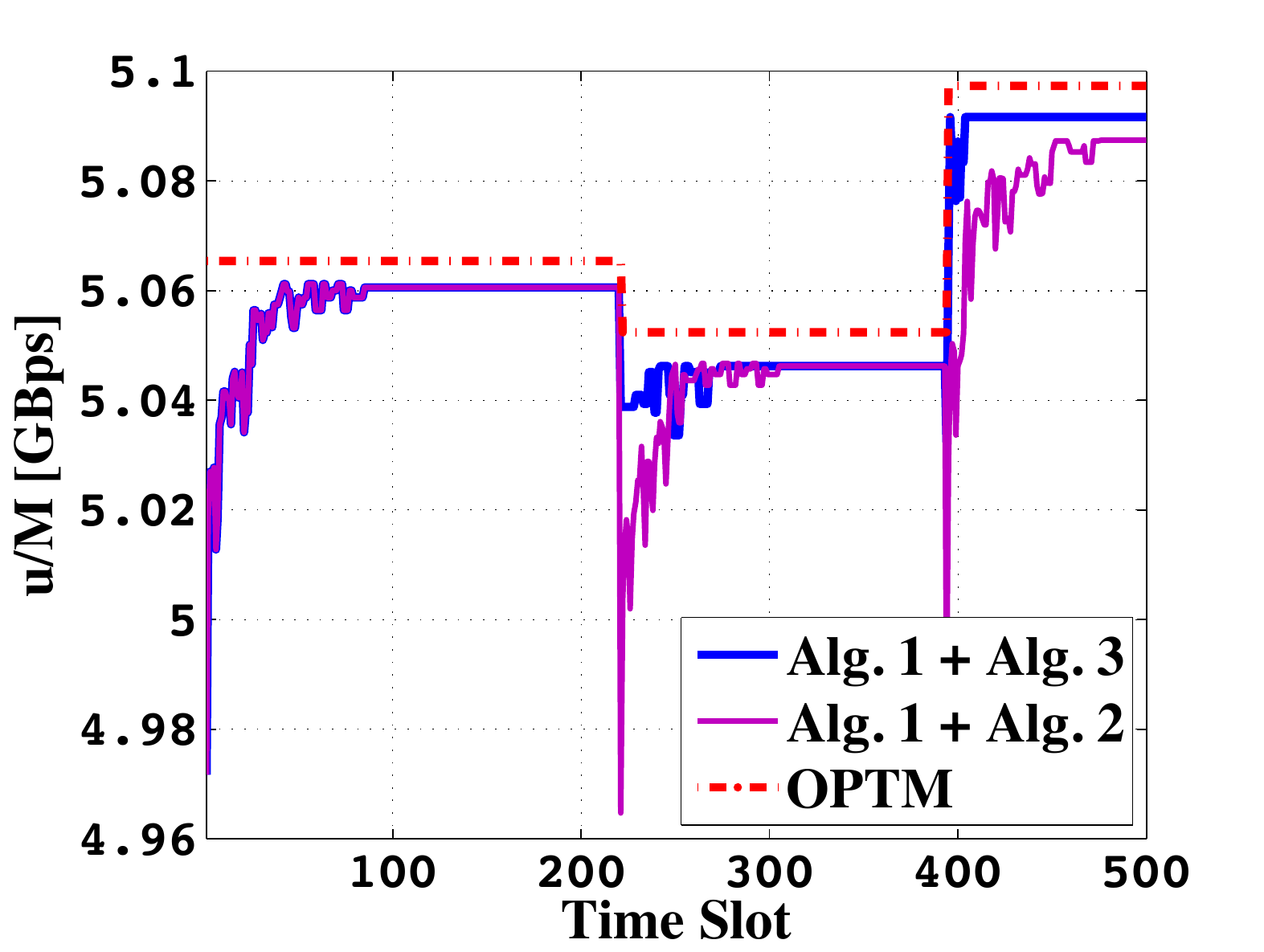}\vspace{-0.15in}
 \caption{The performance of the proposed algorithms. 10 clients and 5 relays join the mmWave network, in which there are 50 clients, 25 relays and 5 APs already. }\label{fig.DYN}\vspace{-0.15in}
\end{figure}

 \begin{figure}[t]\vspace{-0.1in}
  \centering
  \subfloat[]{\hspace{-0.1in}
    \centering
    \includegraphics[width = 0.25\textwidth]{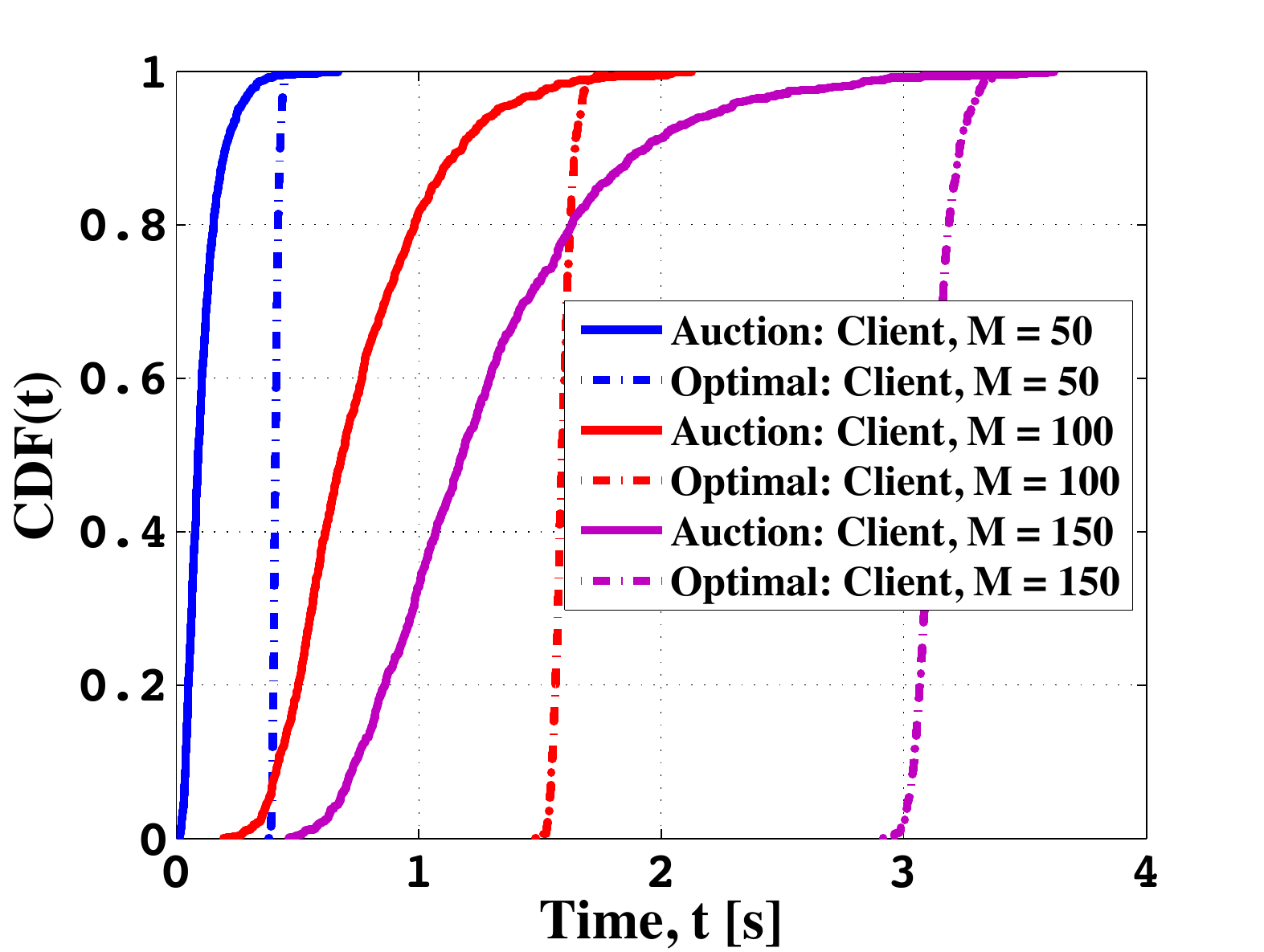}
    \label{fig.TCC_CDF2C}
    }
  \subfloat[]{
    \centering
    \includegraphics[width = 0.25\textwidth]{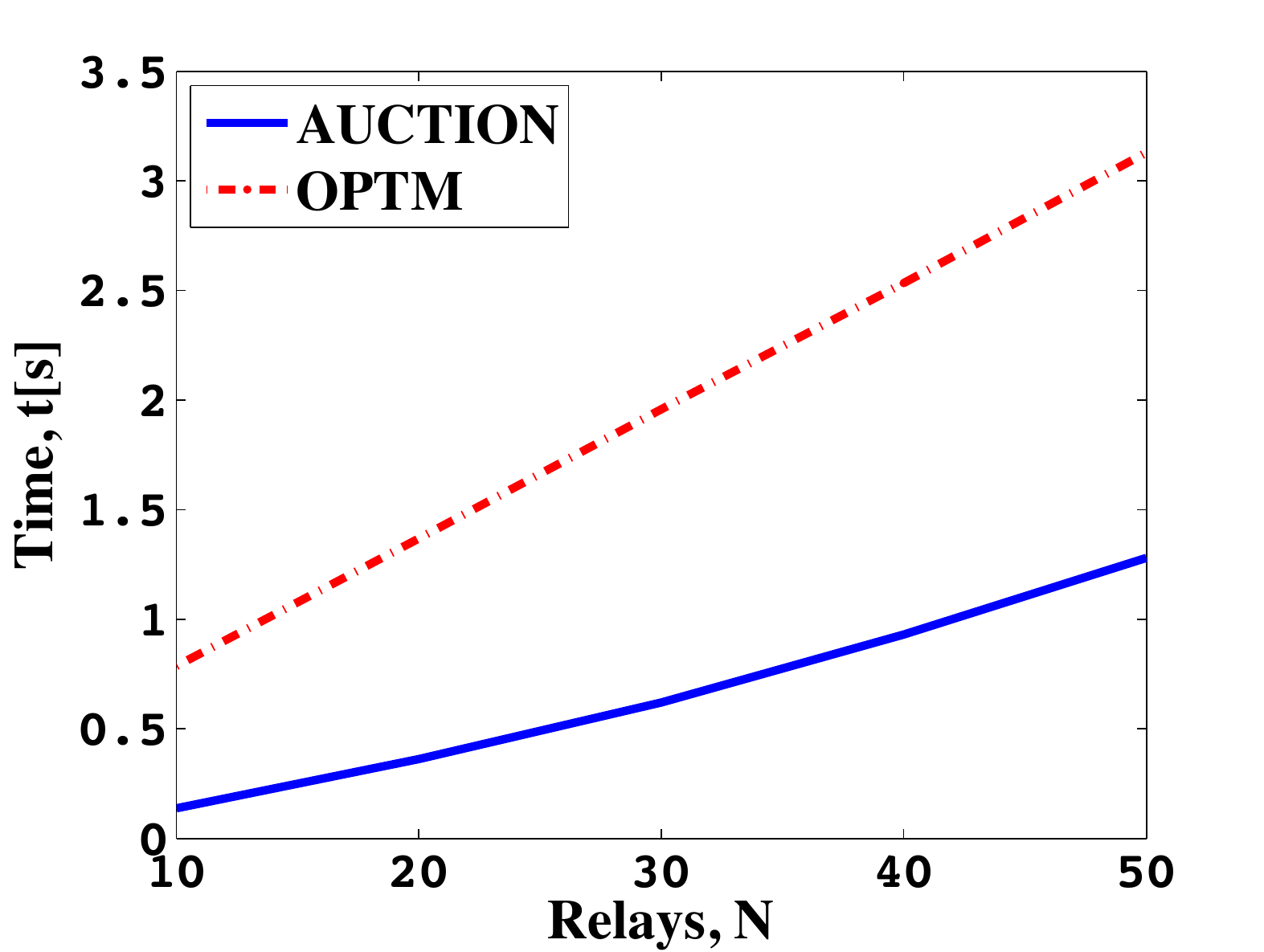}
    \label{fig.TCR_5AP250Client}
    }
  \vspace{-0.1in}\caption{Execution time for Algorithms~\ref{alg_d_c},~\ref{alg_d_n}, and~\ref{alg_d_r_n}. (a) Empirical CDF plots of total CPU time with 10 APs and 25 relays; (b) Average time to converge or find the optimal solution with 10 APs and 150 clients.}
  \label{fig.TC}\vspace{-0.15in}
\end{figure}
 \bibliographystyle{IEEEtran}
 \bibliography{ref}

\end{document}